\documentclass[12pt]{article}
\usepackage{amsfonts,epsfig,latexsym,amscd,amsmath,theorem,mathrsfs}

\newcommand{\ee}{\mathbb{E}} 
 
\newcommand{\D}{\mathbb{D}}
\newcommand{\nv}{{\bf n}}

\newcommand{\Xv}{{\bf X}}

\newcommand{\g}{\mathfrak{g}}

\newcommand{\R}{{\mathbb{R}}}
\newcommand{\Z}{{\mathbb{Z}}}

\newcommand{\I}{{\mathbb{I}}}

\newcommand{\beq}{\begin{equation}}
\newcommand{\eeq}{\end{equation}}
\newcommand{\bea}{\begin{eqnarray}}
\newcommand{\eea}{\end{eqnarray}}
\newcommand{\ben}{\begin{eqnarray*}}
\newcommand{\een}{\end{eqnarray*}}
\newcommand{\ra}{\rightarrow}
\newcommand{\rhu}{\rightharpoonup}
\newcommand{\wra}{\rhu}
\newcommand{\hra}{\hookrightarrow}

\newcommand{\cd}{\partial}
\newcommand{\wt}{\widetilde}
\newcommand{\wh}{\widehat}
\newcommand{\less}{\backslash}

\newcommand{\hess}{{\sf Hess}}

\def \d{\mathrm{d}}

\newcommand{\ip}[1]{\langle #1 \rangle}
\newcommand{\ignore}[1]{}
\newcommand{\downsize}[1]{{\tiny #1}}

\renewcommand{\Pr}{\mathrm{Pr}}

\newcommand{\vol}{{\rm vol}}

\newcommand{\ol}{\overline}

\newcommand{\yvec}{\mbox{\boldmath{$y$}}}

\newcommand{\tr}{{\rm tr}\, }

\newcommand{\eps}{\varepsilon}
\renewcommand{\phi}{\varphi}
\renewcommand{\o}{0}

\theoremstyle{plain}
\newtheorem{thm}{Theorem}

\newtheorem{prop}[thm]{Proposition}

{\theorembodyfont{\rmfamily}
\newtheorem{defn}[thm]{Definition}

}

\newcommand{\news}{\setcounter{equation}{0}}
\newenvironment{proof}{\noindent{\it Proof:\, }}{\hfill$\Box$\vspace*{0.5cm}
}

\renewcommand{\theequation}{\thesection.\arabic{equation}}

\begin{document}

\title{Solitons on tori and soliton crystals}
\author{
J.M. Speight\thanks{E-mail: {\tt speight@maths.leeds.ac.uk}}\\
School of Mathematics, University of Leeds\\
Leeds LS2 9JT, England}

\date{}
\maketitle

\begin{abstract}
Necessary conditions for a soliton on a torus $M=\R^m/\Lambda$ to be a soliton crystal, that is, a spatially periodic array of 
topological solitons in stable equilibrium, are derived. The
stress tensor of the soliton must be $L^2$ orthogonal to $\ee$, the space of parallel symmetric bilinear
forms on $TM$, and, further, a certain symmetric bilinear form on $\ee$, called the hessian, must be positive.
It is shown that, for baby Skyrme models, the first condition actually implies the second. It is also shown that,
for any choice of period lattice $\Lambda$, there is a baby Skyrme model which supports a soliton crystal of periodicity $\Lambda$.
For the three-dimensional Skyrme model, it is shown that any soliton solution on a cubic lattice which
satisfies a virial constraint and is equivariant with respect to (a subgroup of) the lattice
symmetries automatically satisfies both tests. This verifies in particular that the celebrated Skyrme
crystal of Castillejo {\it et al.}, and Kugler and Shtrikman, passes both tests.
\end{abstract}

\maketitle

\section{Introduction}
\news

There are many studies in the mathematical physics literature in which a 
nonlinear field theory known to support topological solitons is studied
not on Euclidean space $\R^m$, but on a torus $T^m=\R^m/\Lambda$. The 
energy minimizers found are then usually interpreted as soliton {\em crystals},
that is, spatially periodic arrays of solitons held in stable equilibrium. 
 However,
once we place the model on a compact domain, every homotopy class of fields will
generically have an energy minimizer. This is true whatever 
period lattice $\Lambda$ we choose, no matter how crazy. Clearly, for a torus
such as the one depicted in figure \ref{fig1}, the energy minimizers cannot
meaningfully be interpreted as soliton crystals: they are an artifact of the
choice of boundary conditions\footnote{Actually, we will see that in the
case of baby Skyrme models, {\em every} two-torus, no matter how bizarre,
does support a soliton crystal for an appropriate choice of potential.}.
The choice of a cubic period lattice certainly {\em looks} more 
plausible. But since we were bound to find an energy minimizer, by compactness,
why should we assume that minimizers found on cubic lattices are not also
artifacts of the boundary conditions? To be sure, we should vary the
energy not just with respect to the field, but also with respect to the
period lattice $\Lambda$. In most numerical studies, this latter variation
(over tori) is only partially performed: the torus is fixed to be cubic, but
its side length is varied. Of
course, there is a good reason for this omission: it is numerically rather intricate
to study field theories on non-rectangular tori. We are aware of only one 
numerical study which systematically does so: Hen and Karliner, in a study
of baby skyrmions, minimized energy over all two-tori \cite{henkar}. 

\begin{figure}
\begin{center}
\includegraphics[scale=0.4]{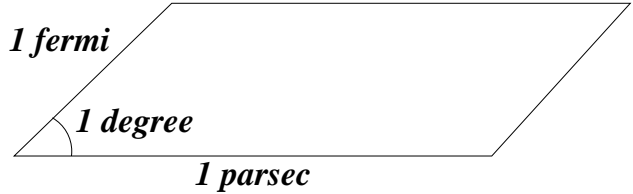}
\end{center}
\label{fig1}
\caption{\sf A possible period lattice for a toric soliton. Clearly a soliton on such a torus is an artifact of the choice of boundary conditions. Or is it?}
\end{figure}

The purpose of this paper is to develop a scheme for determining whether
a given energy minimizer on $\R^m/\Lambda$, assumed to minimize energy
among all fields in its homotopy class on this fixed torus, also minimizes
energy locally with respect to variations of the period lattice. The 
key idea
is that varying the torus among all flat tori with Euclidean metric
is equivalent to fixing the torus but varying the {\em metric}
on the torus among all constant coefficient metrics. This manoeuvre allows
us to formulate the energy variation with respect to the torus
in terms of the stress tensor of the field. We compute both the
first and second variation formulae to give two criteria, involving the stress
tensor, for the field to be a critical point of energy with respect
to variations of the lattice (the first variation) and then, further, a local
minimum of energy with respect to such variations (the second variation). 
The first criterion is that the stress tensor should be $L^2$ orthogonal
to the $\frac12m(m+1)$ dimensional space $\ee$ of 
parallel 
symmetric $(0,2)$ tensors on $\R^m/\Lambda$. The second criterion
is that a certain symmetric bilinear form on $\ee$, called the hessian, should be positive.
These criteria, which  can easily be checked numerically, give necessary,
but not sufficient, conditions for a soliton on a torus to be a genuine 
physical soliton
crystal. (There are two subtleties. First, varying $\Lambda$ locally does not
account for the possibility that energy could be lowered by period-increasing
variations, in which the field is reinterpreted as living on a larger
torus obtained by gluing neighbouring tori together. Such variations do
not just
jump discontinuously in the space of lattices, they also jump
to another homotopy class of fields, and so are inherently inaccessible to
variational calculus. Second, our scheme considers variations of the 
field and the lattice separately, so that cross terms in the second variation,
arising from simultaneous variations of both field and lattice, are
not considered.) Nonetheless, for the sake of
terminological convenience, we shall say that an energy minimizer
on a fixed torus is a soliton {\em lattice} if it satisfies the first 
criterion, and a soliton {\em crystal} if it also satisfies the second.

The paper is structured as follows. In section \ref{sec:var} we derive
the criteria for a general scalar field theory and determine how
symmetries of the field imply symmetries of its stress tensor and hessian. 
In section \ref{sec:bs}
we apply the criteria in the context of the baby-Skyrme model. We find that
a toric baby skyrmion is a soliton lattice if it satisfies a virial
constraint of Derrick type and is conformal on average, in the $L^2$
sense. We show further that every baby skyrmion lattice is a soliton
crystal, that is, the second criterion follows immediately from the
first in the baby Skyrme case. We also show that for any choice of 
period lattice $\Lambda$, there exists a potential for which the baby Skyrme
model has a crystal with this periodicity. 
Throughout this paper our focus is on variation of the energy with respect to the period lattice, rather than the more usual problem of varying the
field on a fixed domain. Indeed, the existence, on any compact domain and in any degree class, of an energy minimizer in a function space with sufficient
regularity for our criteria to make rigorous sense is typically already known in the literature. 
An exception to this would appear to be the baby Skyrme model, whose general existence problem on compact domains does not seem to have been
studied rigorously. At the end of section \ref{sec:bs} we fill in this gap, proving an existence result on general tori for the model with
target space $S^2$ and an arbitrary potential.
In section \ref{sec:skyrme} we
consider the usual nuclear Skyrme model in three dimensions, and some 
interesting variants which are of current phenomenological interest. 
It is shown that,
on a cubic lattice, any energy minimizer which satisfies a certain virial
constraint of Derrick type, and is equivariant with repect to (a 
certain subgroup of) the lattice symmetries is automatically a soliton
crystal. This result implies that, in particular,  the ``Skyrme crystal'' found numerically
by Castillejo {\it et al.} 
\cite{casjonjacverjac} and Kugler and Shtrikman \cite{kugsht}
is a soliton crystal according to our definition. Some concluding remarks
are presented in section \ref{sec:conc}.

\section{Varying over the space of tori}
\news
\label{sec:var}

Consider a general static scalar field theory defined by some energy
functional $E(\phi)$ for a field $\phi:\R^m\ra N$, where $N$ is some
target space. Given an energy minimizer $\phi:\R^m/\Lambda_*\ra N$,
where $\Lambda_*=\{n_1\Xv_1+n_2\Xv_2+\cdots+n_m\Xv_m\: :\: \nv\in\Z^m\}$ is
some fixed lattice in $\R^m$, when can the lifted map $\R^m\ra N$ be interpreted
as a soliton crystal? The answer is that it should be critical, and in
fact stable, with respect to variations of the lattice $\Lambda$
around $\Lambda_*$, as well
as variations of the field. Now, all $m$-tori are diffeomorphic 
through linear maps $\R^m\ra\R^m$, so we can identify them all with 
$M=\R^m/\Lambda_*$, the torus of interest. So the manifold and $\phi:M\ra N$
are now fixed, and varying the lattice is equivalent to varying 
the {\em metric} on $M$ by pulling back the standard Eucldean metric
on $\R^k/\Lambda$ to $M$ by the inverse of the diffeomorphism
$\R^k/\Lambda\ra M$. Let us denote this metric on $M$ by $g_\Lambda$.
Varying among all lattices, one sees that
\beq
g_\Lambda=\sum_{ij}g_{ij}(\Lambda)dx_idx_j
\eeq
where $g_{ij}(\Lambda)$ are constant and $g_{ij}(\Lambda_*)=\delta_{ij}$.
Consider now a curve $g_t$ in this space of metrics on $M$ such that
$g_0=g_{\Lambda_*}$ the Euclidean metric, and denote by 
\beq
\eps=\cd_t|_{t=0}g_t\in\Gamma(T^*M\odot T^*M)
\eeq
its initial tangent vector. Then $\eps$ lies in the space of allowed
variations
\beq\label{eedef}
\ee=\{\sum_{ij}\eps_{ij}dx_idx_j\: :\: \mbox{$\eps_{ij}$ constant, $\eps_{ij}=\eps_{ji}$}\}.
\eeq
This is a $m(m+1)/2$ dimensional subspace of the space of sections of the
rank $m(m+1)/2$ vector bundle $T^*M\odot T^*M$ (where $\odot$ denotes symmetrized tensor product), which is canonically 
isomorphic to any fibre. The canonical isomorphism $\ee\ra T_x^*M\odot T_x^*M$
is given by evaluation $\eps\mapsto \eps(x)$.  Now each fibre $T_x^*M\odot T_x^*M$ has a canonical inner product,
\beq
\ip{\wh\eps,\eps}=\sum_{i,j}\wh\eps(e_i,e_j)\eps(e_i,e_j)
\eeq
where $\{e_1,\ldots,e_m\}$ is any orthonormal frame of vector fields on $M$. Hence, the isomorphism $\ee\ra T_x^*M\odot T_x^*M$
equips $\ee$ with a
canonical inner product, which we
will denote $\ip{\cdot,\cdot}_{\ee}$. Note that this is independent of the choice of base point $x$. The inverse isomorphism is defined
by parallel propagation, so we refer to $\ee$ as the space of parallel
symmetric bilinear forms. 

Now for {\em any} variation $g_t$ of the metric,
\beq
\left.\frac{d E(\phi,g_t)}{dt}\right|_{t=0}=:\ip{\eps,S}_{L^2}
\eeq
where $S$, by definition, is the {\em stress tensor} of the field, defined in analogy with the stress-energy-momentum tensor familiar from relativity
theory (see \cite{baieel} for the original derivation of $S$ in the important case that $E$ is the Dirichlet energy). Like $g$ 
and $\eps$, $S$ is
a section of $T^*M\odot T^*M$. So $E$ is critical for variations of the
lattice $\Lambda$ if and only if $S\perp_{L^2}\ee$. We can reformulate this condition as follows. Let $\ee_0$
denote the orthogonal complement of $g$ in $\ee$, that is, the space of {\em traceless} parallel symmetric
bilinear forms. Then $S\perp_{L^2}\ee$ if and only if
\beq\label{ijd}
\ip{S,g}_{L^2}=0
\eeq
and
\beq\label{injd}
S\perp_{L^2}\ee_0.
\eeq
Condition (\ref{ijd}) is a virial constraint, analogous to the constraint obtained for 
the model on $\R^m$ using the Derrick scaling argument \cite{der}. In fact, if we replace the torus $M$ by
$\R^m$ and assume that $\phi$ is a finite energy solution, then (\ref{ijd}) coincides precisely
with the Derrick virial constraint.
Similarly (\ref{injd}) also coincides (under the replacement of $M$ by $\R^m$) with the collection of
virial constraints ``beyond Derrick's theorem'' obtained by Manton \cite{man-der} and generalized recently
by Domokos {\it et al.} \cite{domhoyson}. In brief, then, in order to be critical with respect to
variations of the period lattice, a toric soliton must satisfy the generalized Derrick virial constraints
in each lattice cell.

A convenient reformulation of (\ref{injd}) arises as follows. Given any point $x\in M$
and any pair of tangent vectors $X,Y\in T_xM$, these have unique parallel
propagations over $M$, which we also denote $X,Y$. Associated to the map
$\phi$, we can define a symmetric bilinear form
\beq\label{arireb}
\Delta:T_xM\times T_xM\ra\R,\qquad
\Delta(X,Y)=-2\int_MS_0(X,Y)\vol_g
\eeq
where $S_0$ denotes the trace-free part of $S$ (the factor of $-2$ is for later convenience).
Using the canonical identification of $\ee$ with any fibre $T_x^*M\odot
T_x^*M$, we can identify $\Delta$ with an element of $\ee$.
Then $S\perp_{L^2}\ee_0$ if and only if $S_0\perp_{L^2}\ee_0$, which holds if and only if $\Delta$ is orthogonal to
$\ee_0$ with respect to $\ip{\cdot,\cdot}_\ee$, and hence, if and only if
\beq
\Delta=\lambda g,
\eeq
for some constant $\lambda\in\R$ which, if required, can be found by evaluating $\Delta$
on any unit vector.

Assume that the pair $(\phi,g)$ satisfies these conditions, (\ref{ijd}), (\ref{injd}).
Then we can define the second variation of $E$
with respect to the metric in the affine space $g+\ee$. Let $g_{s,t}$
be a two-parameter family of metrics in $g+\ee$ with $g_{0,0}=g$, and
let $\wh\eps=\cd_sg_{s,t}|_{(0,0)},\eps=\cd_tg_{s,t}|_{(0,0)}\in\ee$. Then
the {\em hessian} of $E:g+\ee\ra\R$ at $g$ is, by definition, the 
symmetric bilinear form $\hess\in\ee^*\odot\ee^*$ such that
\beq
\hess(\wh\eps,\eps)=\left.\frac{\cd^2 E(\phi,g_{s,t})}{\cd s\cd t}\right|_{s=t=0}.
\eeq
The pair $(\phi,g)$ is a local minimum of $E$, with respect to variations
of the lattice, if this bilinear form is positive definite. 
This leads us
to the following:

\begin{defn}\label{ariel}
 An $E$ minimizer $\phi:\R^m/\Lambda\ra N$ is a {\em soliton lattice} if
its stress tensor $S(\phi)$ is $L^2$ orthogonal to $\ee$, the space of 
parallel symmetric bilinear forms. A soliton lattice is a {\em soliton crystal} if, in addition,
its hessian is positive definite.
\end{defn}

The detailed structure of $S$, and hence of $\hess$, depend on the details of the 
field theory, We can, however, find a semi-explicit formula for $\hess$ which
turns out to be rather useful for our purposes. To state it, we need to define
the natural contraction for pairs of bilinear forms on $M$. So, let $A,B$ be
bilinear forms (i.e.\ $(0,2)$ tensors) on $M$ and $e_1,\ldots,e_m$ be any
local orthonormal frame of vector fields on $M$. Then we define $A\cdot B$
to be the bilinear form
\beq
(A\cdot B)(X,Y)=\sum_i A(X,e_i)B(e_i,Y).
\eeq
If we identify a bilinear form with its matrix of components relative to the
frame $\{e_i\}$ then this contraction coincides with matrix multiplication.

\begin{prop}\label{hessprop}
 Let $\phi$ be a soliton lattice, and $\hess$ be its hessian. Then, for
all $\wh\eps,\eps\in\ee$,
$$
\hess(\wh\eps,\eps)=\ip{\dot S,\eps}_{L^2}-2\ip{\wh\eps,S\cdot\eps}_{L^2},
$$
where $\dot S=\cd_s|_{s=0}S(\phi,g_s)\in\Gamma(T^*M\odot T^*M)$ for 
any generating curve $g_s$ for $\wh\eps$. In particular, if $\wh\eps=\eps$, then
$$
\hess(\eps,\eps)=\ip{\dot S,\eps}_{L^2}.
$$
\end{prop}

\begin{proof} Let $g_{s,t}$ be any two-parameter variation of $g=g_{0,0}$ in
$g+\ee$, and $\wh\eps=\cd_s|_{s=0}g_{s,0}$, $\eps_s=\cd_t|_{t=0}g_{s,t}$,
$\eps=\eps_0$. 
Let $g_s=g_{s,0}$.
Then
\bea
\hess(\wh\eps,\eps)&=&\left.\frac{d\:}{ds}\right|_{s=0}
\left.\frac{\cd E(\phi,g_{s,t})}{\cd t}\right|_{t=0}
=\left.\frac{d\:}{ds}\right|_{s=0}\ip{S(\phi,g_s),\eps_s}_{L^2,g_s}\nonumber \\
&=&\left.\frac{d\:}{ds}\right|_{s=0}\int_M
\ip{S(\phi,g_s),\eps_s}_{g_s}\vol_{g_s}\nonumber \\
&=&\int_M\left\{
\ip{\dot{S},\eps}\vol_g+\ip{S,\dot\eps}\vol_g
+\left.\frac{d\:}{ds}\right|_{s=0}\ip{S,\eps}_{g_s}\vol_g
+\ip{S,\eps}_g\left.\frac{d\:}{ds}\right|_{s=0}\vol_{g_s}\right\}.
\label{flita}
\eea
Now $\dot\eps=\cd_s|_{s=0}\eps_s\in\ee$ and $S\perp_{L^2}\ee$ by assumption
($\phi$ is assumed to be a lattice), so the second term vanishes, upon 
integration over $M$. 

For any fixed pair $A,B$ of symmetric bilinear forms,
\beq
\left.\frac{d\:}{ds}\right|_{s=0}\ip{A,B}_{g_s}
=-2\ip{A\cdot B, \wh\eps}_g.
\eeq
To check this, we can work in a local coordinate chart and use the Einstein
summation convention. Since $g_{ij}(s)g^{jk}(s)=\delta_i^k$ for all $s$, we deduce
that $\dot{g}^{ij}=d g^{ij}(s)/ds|_{s=0}=-g^{ip}\eps_{pq}g^{qj}$. Hence
\bea
\left.\frac{d\:}{ds}\right|_{s=0}\ip{A,B}_{g_s}
&=&
\left.\frac{d\:}{ds}\right|_{s=0}A_{ij}g^{jk}(s)B_{kl}g^{li}(s)\nonumber \\
&=&-A_{ij}g^{jp}\eps_{pq}g^{qk}B_{kl}g^{li}
-A_{ij}g^{jk}B_{kl}g^{lp}\eps_{pq}g^{pi}\nonumber \\
&=&-2(A\cdot B)_{il}g^{lp}\eps_{pq}g^{qi}
=-2\ip{A\cdot B,\eps}_g.
\eea
Hence the third term in (\ref{flita}) reproduces the second term in the
formula claimed.

The variation of the volume form is known to be \cite[p.~82]{baiwoo}
\beq\label{sopstr}
\left.\frac{d\:}{ds}\right|_{s=0}\vol_{g_s}=\frac12\ip{\wh\eps,g}\vol_g.
\eeq
Note that, since $\wh\eps,g\in\ee$, $\ip{\wh\eps,g}=\ip{\wh\eps,g}_\ee$,
which is constant. Hence, the last term reduces to $\ip{\wh\eps,g}_\ee
\ip{S,\eps}_{L^2}$ which vanishes since $\phi$ is a lattice. This completes
the proof of the first formula.

Now assume $\wh\eps=\eps$. For any triple of symmetric bilinear forms $A,B,C$,
$
\ip{A,B\cdot C}_g=\ip{C,A\cdot B}_g
$
and so $\ip{\eps,S\cdot\eps}_{L^2}=\ip{S,\eps\cdot\eps}_{L^2}=0$
since, for all $\eps\in\ee$, $\eps\cdot\eps\in\ee$, and $S\perp_{L^2}\ee$.
\end{proof}

We will see that, in the case of three-dimensional Skyrme models, the task of checking that a given $E$-minimizer
$\phi$ satisfies Definition \ref{ariel} can be greatly simplified if $\phi$ is equivariant with respect to (some
subgroup of) the symmetries of the period lattice $\Lambda$. To exploit equivariance we first need to extract its consequences
for $S$ and $\hess$. The following basic symmetry properties may prove useful in contexts other than toric solitons, so we
formulate them in some generality.

Let $K$ be any group acting isometrically on the left
on the riemannian manifold $(M,g_0)$, and let $E(\phi,g)$ be a geometrically natural energy functional
on the space of smooth maps $\phi:M\ra N$ and metrics on $M$. By geometrically natural, we mean that, for any diffeomorphism
$\psi:M\ra M$, 
\beq
E(\phi\circ\psi,g)=E(\phi,(\psi^{-1})^*g)
\eeq
for all $\phi$ and $g$. In local coordinates, $\psi$ can be thought of as a passive transformation, that is, a change of coordinate, 
in which case, the condition above simply means that
$E$ is independent of the choice of coordinates. Note that $E(\phi,g_0)$ (with fixed metric) is automatically invariant under the $K$ action on $M$, since
$E(\phi\circ k^{-1},g_0)=E(\phi,k^*g_0)=E(\phi,g_0)$ as $K$ acts by isometries.
Let $K$ also act on $N$ on the left, in such a way that $E$ is invariant under this $K$ action, that is,
\beq
E(k\circ\phi,g)=E(\phi,g)
\eeq
for all $k\in K$, $\phi$ and $g$. Such a functional will be called $K$-invariant.
We say that a field $\phi:M\ra N$ is $K$-equivariant if $\phi\circ k=k\circ\phi$ for all $k\in K$. 

\begin{prop}
\label{sym1}
 Let $\phi:(M,g_0)\ra N$ be a $K$-equivariant field, and $S$ be its stress tensor with respect to a geometrically natural
$K$-invariant energy functional. Then $k^*S=S$ for all $k\in K$. 
\end{prop}

\begin{proof} 
Let $g_t$ be an arbitrary variation of $g_0$, and $\eps=\cd|_{t=0}g_t$. Since $K$ acts isometrically on $(M,g_0)$, each $k^*$ is $L^2$ self-adjoint
on $\Gamma(T^*M\odot T^*M)$. Hence
\bea
\ip{k^*S,\eps}_{L^2}&=&\ip{S,k^*\eps}_{L^2}
=\left.\frac{d\: }{dt}\right|_{t=0}E(\phi,k^*g_t)\nonumber\\
&=&\left.\frac{d\: }{dt}\right|_{t=0}E(\phi\circ k^{-1},g_t)\qquad\mbox{($E$ is geometrically natural)}\nonumber \\
&=&\left.\frac{d\: }{dt}\right|_{t=0}E(k^{-1}\circ\phi,g_t)\qquad\mbox{($\phi$ is $K$-equivariant)}\nonumber \\
&=&\left.\frac{d\: }{dt}\right|_{t=0}E(\phi,g_t)\qquad\mbox{($E$ is $K$-invariant)}\nonumber \\
&=&\ip{S,\eps}_{L^2}.
\eea
Since this holds for all variations, and all $k$, the result follows.
\end{proof}

A similar argument shows that the hessian is also $K$-invariant. 

\begin{prop}\label{sym2}
 Let $\phi:(M,g_0)\ra N$ be a $K$-equivariant field, and $S$ be its stress tensor with respect to a geometrically natural
$K$-invariant energy functional $E$. Assume that $S$ is $L^2$ orthogonal to some $K$-invariant linear subspace $\ee\subset \Gamma(T^*M\odot T^*M)$, and let
$\hess:\ee\times\ee\ra\R$ 
be the hessian of $E$ with respect to variations in the affine space $g_0+\ee$. Then $(k^*)^*\hess=\hess$ for all $k\in K$.
\end{prop}

\begin{proof} Let $g_{s,t}$ be an arbitrary two-parameter variation of $g_0=g_{0,0}$, $\wh\eps=\cd_s|_{s=0}g_{s,0}$, $\eps=\cd_t|_{t=0}g_{0,t}$, and $k\in K$.
Then
\bea
(k^*)^*\hess(\wh\eps,\eps)&=& \hess(k^*\wh\eps,k^*\eps)=\left.\frac{\cd^2\: \: }{\cd s\cd t}\right|_{s=t=0}E(\phi,k^*g_{s,t})\nonumber \\
&=&\left.\frac{\cd^2\: \: }{\cd s\cd t}\right|_{s=t=0}E(\phi\circ k^{-1},g_{s,t})\qquad\mbox{($E$ is geometrically natural)}\nonumber \\
&=&\left.\frac{\cd^2\: \: }{\cd s\cd t}\right|_{s=t=0}E(k^{-1}\circ \phi,g_{s,t})\qquad\mbox{($\phi$ is $K$-equivariant)}\nonumber \\
&=&\left.\frac{\cd^2\: \: }{\cd s\cd t}\right|_{s=t=0}E(\phi,g_{s,t})\qquad\mbox{($E$ is $K$-invariant)}\nonumber \\
&=&\hess(\wh\eps,\eps).
\eea
Since this holds for all variations, and all $k$, the result follows.
\end{proof}

To conclude this section, we return to the setting of interest: $M=\R^m/\Lambda$, a torus with the euclidean metric $g$, 
$\ee$ is the space of
parallel symmetric bilinear forms on $M$, and $K$ is some subgroup of $O(m)$ which preserves $\Lambda$ (and hence
acts isometrically on $M$). Note that $\ee$ is indeed $K$-invariant. Let $\phi$ be $K$-equivariant, and $\Delta\in\ee$ be the bilinear form 
defined in (\ref{arireb}). Then $\Delta$ is itself $K$-invariant.

\begin{prop}\label{sym3}
 If $\phi:M\ra N$ is $K$-equivariant, then $k^*\Delta=\Delta$ for all $k\in K$.
\end{prop}

\begin{proof} Clearly
\beq
\Delta(X,Y)=\lambda g(X,Y)-2\int_M S(X,Y)\vol_g=:\lambda g(X,Y)+\Delta'(X,Y)
\eeq
for some constant $\lambda$, and $k^*g=g$ for all $k\in K$ since $K$ acts isometrically. Hence, it suffices to show that
$k^*\Delta'=\Delta'$. But this follows immediately from Proposition \ref{sym1}.
\end{proof}

\section{Baby skyrmion crystals}
\news
\label{sec:bs}

The baby Skyrme model concerns a single scalar field $\phi:M\ra N$ where
$(M,g)$ is an oriented riemanian two-manifold (in our case, a torus
$\R^2/\Lambda$) and $(N,h,\omega)$ is a compact k\"ahler manifold 
(usually chosen to be
$S^2$) with metric $h$ and k\"ahler form $\omega$. The energy is
\beq
E(\phi)=\frac12\int_M\left\{|\d\phi|^2+|\phi^*\omega|^2+U(\phi)^2\right\}
\vol_g
\eeq
where we have found it convenient to write the potential as $\frac12U^2$ where
$U:N\ra\R$ is some function. It is conventional to label the three terms in $E$ as $E_2$, $E_4$ and $E_0$
respectively. The subscript specifies the degree of the integrand
thought of as a polynomial in spatial partial derivatives.

Let $\phi:M=T^2/\Lambda\ra N$
minimize $E$ among maps (on this fixed torus) in its homotopy class. Under
what circumstances is $\phi$ a soliton lattice, as defined in section
\ref{sec:var}?
The stress tensor of $\phi$ is easily computed \cite{jayspesut},
\beq\label{bsS}
S(\phi,g)=\frac14\left(|\d\phi|^2_g-|\phi^*\omega|^2_g+U(\phi)^2
\right)g-\frac12\phi^*h.
\eeq
Now $\phi$ is a lattice if and only if $S(\phi,g)$ is $L^2$ orthogonal to
both $g$ and $\ee_0$,
the space of traceless parallel symmetric bilinear forms, yielding conditions (\ref{ijd}) and (\ref{injd}). Now $\ip{g,g}=2$ and 
\beq
\ip{g,\phi^*h}=\sum_{i,j}g(e_i,e_j)\phi^*h(e_i,e_j)=\sum_i\phi^*h(e_i,e_i)
=|d\phi|^2
\eeq
so
\beq
\ip{S,g}_{L^2}=\frac12(E_0-E_4).
\eeq
Hence (\ref{ijd}) becomes the familiar virial constraint
\beq\label{ijdbs}
E_0=E_4.
\eeq
 This is the condition
which is enforced in numerical studies of baby skyrmion crystals by minimization
of $E$ over the sidelength of the torus.

For all $\eps\in\ee_0$,
$\ip{S,\eps}_{L^2}=-\frac12\ip{\phi^*h,\eps}_{L^2}$ since $\eps$ is pointwise
orthogonal to $g$. In two dimensions, $\ee_0$ is spanned by
\beq
\eps_1=dx_1^2-dx_2^2,\qquad \eps_2=2dx_1dx_2.
\eeq
So (\ref{injd}) is equivalent to 
\bea
\ip{\phi^*h,\eps_1}_{L^2}&=&\int_M\left(\left|\d\phi\frac{\cd\: }{\cd x_1}\right|^2
-\left|\d\phi\frac{\cd\: }{\cd x_2}\right|^2\right)\vol_g=0\\
\label{injdbs}
\ip{\phi^*h,\eps_2}_{L^2}&=&\int_M h(\d\phi\frac{\cd\: }{\cd x_1},
\d\phi\frac{\cd\: }{\cd x_2})\vol_g=0.
\eea
Note that these conditions, like the virial constraint, are easy to check numerically.
If the latter two conditions hold pointwise, rather than just as
integrals, this means precisely that the mapping $\phi$ is (weakly) conformal.
So a natural way to describe the condition $S\perp_{L^2}\ee_0$ is that the
mapping $\phi$ must be {\em conformal on average} (in the $L^2$ sense).
We can reformulate this condition using the symmetric bilinear form $\Delta$,
defined in (\ref{arireb}). In this case, $-2S_0=\phi^*h$, so $\phi$ is
conformal on average if and only if
\beq
\Delta=\int_M\phi^*h\vol_g=\lambda g
\eeq
where $\lambda$ is some constant. Taking the trace of both sides, one sees that $\lambda=E_2$.
 In abbreviated notation, then, we
may express the conditions for $\phi$ to be a soliton lattice as
\beq
E_0=E_4,\qquad \Delta=\int_M\phi^*h\vol_g=E_2g.
\eeq
This result was already derived, via a slightly different argument, in
\cite{jayspesut}.

So, given an $E$ minimizer $\phi$ 
which is conformal on average and satisfies $E_0=E_4$,
we know that $E$ is critical with respect to variations of the period lattice.
To check if, in addition, it is {\em stable} with respect to variations of the 
period lattice, we need to show that its hessian is positive definite. In
fact this second step turns out to be redundant, as we now show.

\begin{prop}\label{latimpcry}
 Let $\phi$ be a baby skyrmion lattice. Then its hessian is
positive definite. Hence every baby skyrmion lattice is a soliton crystal.
\end{prop}

\begin{proof} 
We are given that $S(\phi)$ is $L^2$ orthogonal to $\ee$.
Let $g_{s,t}$ be a two-parameter variation of $g=g_{0,0}$ in $g+\ee$, and
$\wh\eps,\eps,g_s,\dot{S}$ be as defined as in Proposition \ref{hessprop}.
From (\ref{bsS}) we see that
\beq
\dot{S}=\lambda g+\frac14(|\d\phi|^2-|\phi^*\omega|^2+U(\phi)^2)\wh\eps.
\eeq
Hence, if $\wh\eps=g$ and $\eps\in\ee_0$ then $\dot{S}=\lambda g$, so
\beq
\hess(g,\eps)=\ip{\lambda g,\eps}_{L^2}-2\ip{g,S\cdot\eps}_{L^2}
=0-2\ip{S\cdot g,\eps}_{L^2}
=-2\ip{S,\eps}_{L^2}=0
\eeq
since $\phi$ is a soliton lattice. Hence it suffices to show that
$\hess(g,g)>0$ and $\hess(\eps,\eps)>0$ for all $\eps\in\ee_0\less\{0\}$.

Now
\beq
\hess(g,g)=\left.\frac{d^2\:}{dt^2}\right|_{t=1}E(\phi,g_t)
\eeq
where $g_t=tg$. Clearly
\beq
E_0(\phi,g_t)=tE_0(\phi,g),\qquad
E_2(\phi,g_t)=E_2(\phi,g),\qquad
E_4(\phi,g_t)=t^{-1}E_4(\phi,g)
\eeq
so
\beq
\hess(g,g)=2E_4>0.
\eeq

Finally, if $\wh\eps=\eps\in\ee_0\less\{0\}$ then, by 
Proposition \ref{hessprop},
\bea
\hess(\eps,\eps)&=&\ip{\dot S,\eps}_{L^2}
=\ip{\lambda g+\frac14(|\d\phi|^2-|\phi^*\omega|^2+U(\phi)^2)\eps,\eps}_{L^2}\nonumber \\
&=&\frac14\ip{\eps,\eps}_\ee\int_M(|d\phi|^2-|\phi^*\omega|^2+U(\phi)^2)\vol_g\nonumber \\
&=&\frac12 E_2\ip{\wh\eps,\eps}_\ee>0
\eea
since $E_0=E_4$.
\end{proof}

Hence, rather remarkably, for baby Skyrme models every soliton lattice is a soliton crystal.
It follows that the baby skyrmion lattices found in \cite{jayspesut}, for example,
for which the conditions (\ref{ijdbs}), (\ref{injdbs}) were checked numerically, are
soliton crystals according to our definition. These were defined on equianharmonic
tori, that is tori with $\Lambda$ generated by $L$ and $Le^{i\pi/3}$. Another remarkable
feature of baby Skyrme models is that {\em every} torus, no matter how bizarre its period
lattice, supports a soliton crystal for an appropriate choice of potential function $U$.
We next outline the construction of this special potential.

We specialize to the case where $N=S^2$, the 
unit sphere with its usual metric and complex structure. There is a well-known
topological energy bound on $E_2(\phi)$ due to Lichnerowicz \cite{lic}
\beq
E_2(\phi)\geq 2\pi \deg(\phi)
\eeq
where $\deg(\phi)$ denotes the degree of the map $\phi:M\ra S^2$. Equality holds
if and only if $\phi$ is holomorphic. Less well-known is that there is also
a topological lower energy bound on $E_0+E_4$ \cite{spe-semicompactons}
\beq
E_0(\phi)+E_4(\phi)\geq 4\pi\ip{U}\deg(\phi)
\eeq
where $\ip{U}$ denotes the average value of the function $U:S^2\ra\R$, with 
equality if and only if
\beq\label{bog}
\phi^*\omega=*U\circ\phi.
\eeq
Let us choose and fix a torus $M=\R^2/\Lambda$ and use
a stereographic coordinate 
\beq
W=\frac{\phi_1+i\phi_2}{1-\phi_3}
\eeq
on $S^2$ and a complex coordinate $z=x_1+ix_2$ on $M$. Then there is a
degree 2 holomorphic map $\phi_\wp:M\ra S^2$ defined by taking
$W(z)=\wp(z)$, the Weierstrass p-function corresponding to lattice $\Lambda$.
Then $W(z)$ satisfies the ordinary differential equation
\beq
W'(z)^2=4W(z)^3-c_2W(z)-c_3
\eeq
where $c_2,c_3$ are constants, depending on $\Lambda$, which are more 
conventionally denoted $g_2,g_3$ \cite[p.~159]{law}. 
Since $\phi_\wp$ is holomorphic,
\beq
\phi_\wp^*\omega=
\frac{4|W'(z)|^2}{(1+|W|^2)}\:\frac{i}{2}\d z\wedge\d\ol{z}.
\eeq
Hence, if we choose the potential function so that
\beq
U(W)=\frac{4|4W^3-c_2W-c_3|}{(1+|W|^2)^2}
\eeq
then the field $\phi_\wp$ 
is holomorphic and satisfies (\ref{bog}) and so
simultaneously minimizes both $E_2$ and $E_0+E_4$ among all degree 2 maps
$M\ra S^2$. Hence, this field is an $E$ minimizer. It follows immediately
from (\ref{bog}) that $E_0=E_4$, and since $\phi_\wp$ is holomorphic, it is weakly 
conformal, and hence certainly conformal on average. Hence $\phi_\wp$
is a soliton lattice and further, by Proposition \ref{latimpcry}, a soliton 
crystal. Note that this construction, which generalizes in obvious fashion
an observation of Ward \cite{war} in the case of a square lattice, works no matter what choice of period
lattice $\Lambda$ we start with, although, of course, the potential
function $U$ depends on $\Lambda$. In all cases, the potential
$V=\frac12 U^2$ has four vacua, located at the critical values of the
p-function associated with $\Lambda$. One of these is always the North pole
(corresponding to $W=\infty$) but the other three, located at roots of the
polynomial $4W^3-c_2W-c_3$ move around as $\Lambda$ is varied. Conversely,
given any choice of four distinct points on $S^2$, one can rotate them so that one 
lies at $(0,0,1)$, then construct a degree
two elliptic function with critical values at precisely those (rotated) points. 
This 
function determines a period lattice $\Lambda$, and a potential function $U$
such that the field $\phi_\wp$ is a soliton crystal for the model with potential
$U$. Of course, these potentials are specially constructed to support exact
soliton crystals, but there are other examples of four-vacuum potentials
outside this class which are known numerically, in light of 
Proposition \ref{latimpcry}, to support soliton crystals with topological 
charge per unit cell equal to two. In all known cases, the period lattice has 
the same geometry (up to scale) as that predicted by the p-function with
critical values at the vacua (though only exceptionally symmetric cases where
the corresponding p-function has square or equianharmonic period lattice
 have been studied). It would be interesting to see whether this is
a general phenomenon, by conducting a thorough numerical analysis along the
lines of Karliner and Hen's \cite{henkar}.

As we have seen, for a given fixed period lattice $\Lambda$, we can reverse engineer a potential $V(\phi)=\frac12U(\phi)^2$ so that the
baby Skyrme model with target space $S^2$ has a smooth (in fact, holomorphic)  energy minimizer of degree 2 on the torus $M=\R^2/\Lambda$.
But what about degree classes $\deg(\phi)\neq 2$, or potentials other than this very special choice? Existence of minimizers of the
baby Skyrme energy on compact domains has not been rigorously studied previously, and is not entirely trivial; for example,
it is known that no degree $\pm1$ minimizer exists on any torus for the pure sigma model (with the potential and Skyrme terms absent). 
Existence of minimizers on $\R^2$ with degree $\pm 1$ for
the potential $V(\phi)=\lambda(1-\phi_3)^2$ with $\lambda>0$ sufficiently small has been established via the
concentration-compactness method by Lin and Yang in \cite{linyan-bs}.  
Their analysis suggests the essential estimate required for our purposes. 

Choose and fix a period lattice $\Lambda$ and let $M=\R^2/\Lambda$. Denote by $L^2$ the space of square integrable functions 
$M\ra\R^3$ and by $H^1$ the subspace of $L^2$ consisting of maps whose first partial derivatives are also in $L^2$. These are Hilbert spaces
with respect to the usual inner products
\beq
\ip{\phi,\psi}_{L^2}=\int_M\phi\cdot\psi,\qquad
\ip{\phi,\psi}_{H^1}=\ip{\phi,\psi}_{L^2}+\ip{\phi_x,\psi_x}_{L^2}+\ip{\phi_y,\psi_y}_{L^2}.
\eeq
For each fixed $k\in\Z$ we define
\beq
X_k=\{\phi\in H^1\: : \: \mbox{$|\phi|=1$ almost everywhere, $\deg(\phi)=k$}\}
\eeq
where 
\beq
\deg(\phi)=\frac{1}{4\pi}\int_M\phi\cdot(\phi_x\times\phi_y).
\eeq
The result below rests on three standard theorems of functional analysis: Alaoglu's Theorem \cite[p.~125]{mor}
(every bounded sequence in a reflexive Banach space, for 
example, a Hilbert space, has a weakly convergent subsequence), Rellich's Lemma \cite[p.~144]{ada} 
(the inclusion $H^1(M)\hra L^2(M)$ is compact for compact $M$) and
Tonelli's Theorem \cite[p.~22]{eva} ($f\mapsto \int_M F(f)$ is sequentially weakly lower semicontinuous on $L^2$ if $F:\R^m\ra\R$ is convex). 

\begin{thm}\label{ooo-analysis}
 Let $V:S^2\ra\R$ be any function admitting a convex extension to $\R^3$. Then, for any $k\in\Z$ the functional
$$
E:X_k\ra [0,\infty],\qquad E(\phi)=\int_M\left(\frac12|\phi_x|^2+\frac12|\phi_y|^2+\frac12|\phi_x\times\phi_y|^2+V(\phi)\right),
$$
attains a minimum.
\end{thm}

\begin{proof} Let $\phi_n\in X_k$ be a minimizing sequence for $E$, that is, $E(\phi_n)\ra \inf_{\psi\in X_k}E(\psi)$. We will
repeatedly extract nested subsequences from $\phi_n$, still denoted $\phi_n$, with various convergence properties. We will
denote strong convergence by $\ra$ and weak convergence by $\wra$, the space concerned being explicitly specified. 

Since $M$ is compact and $|\phi_n|=1$ almost everywhere, $\|\phi_n\|_{H^1}^2\leq 2E(\phi_n)+{\rm Vol}(M)$, is bounded. Hence, by Alaoglu's Theorem,
there is a subsequence $\phi_n$ and $\phi\in H^1$ such that $\phi_n\stackrel{H^1}{\wra}\phi$. Again, since $M$ is compact, the inclusion
$\iota:H^1\hra L^2$ is compact by Rellich's Lemma, so $\phi_n$, a bounded sequence in $H^1$,
 has a subsequence converging strongly in $L^2$, and hence weakly in $L^2$. By
uniqueness of (weak) limits, its limit must be $\phi$, that is $\phi_n\stackrel{L^2}{\ra}\phi$. Now $|\phi_n|=1$ 
 almost everywhere, so on any open set $\Omega\subset M$, $\int_\Omega(1-|\phi|^2)=\lim\int_\Omega(1-|\phi_n|^2)=0$, so $|\phi|=1$ almost everywhere also.
Hence, we can replace $V$ in the formula for $E$ by its convex extension to $\R^3$. Then every term in the integrand of $E$ is convex, so, 
by Tonelli's Theorem, $E:H^1\ra[0,\infty]$ is a sum of sequentially weakly lower semicontinuous functionals, and hence is itself
sequentially weakly lower semicontinuous. Hence $E(\phi)\leq \lim E(\phi_n)=\inf_{\psi\in X_k}E(\psi)$. It remains to show that $\phi\in X_k$,
that is, $\deg(\phi)=k$.

Both $\phi$ and (since $E(\phi)$ is finite) $\phi_x\times\phi_y$ are in $L^2$, so $\deg(\phi)$ exists. Furthermore,
\bea
4\pi|\deg(\phi)-k|&=&4\pi|\deg(\phi)-\deg(\phi_n)|
=\left|\int_M\phi\cdot(\phi_x\times\phi_y)-\phi_n\cdot(\cd_x\phi_n\times\cd_y\phi_n)\right|\nonumber \\
&\leq&|\ip{\phi,\phi_x\times\phi_y-\cd_x\phi_n\times\cd_y\phi_n}_{L^2}|+|\ip{\phi-\phi_n,\cd_x\phi_n\times\cd_y\phi_n}_{L^2}|\nonumber \\
&\leq&|\ip{\phi,\phi_x\times\phi_y-\cd_x\phi_n\times\cd_y\phi_n}_{L^2}|+\|\phi-\phi_n\|_{L^2}\|\cd_x\phi_n\times\cd_y\phi_n\|_{L^2}
\eea
Now $\|\cd_x\phi_n\times\cd_y\phi_n\|_{L^2}$ is bounded (by $2E(\phi)+1$, for example) and hence, by Alaoglu's Theorem, $\cd_x\phi_n\times\cd_y\phi_n$ has
a subsequence converging weakly in $L^2$. By uniqueness of weak limits, its limit must be $\phi_x\times\phi_y$, so 
$\phi_x\times\phi_y-\cd_x\phi_n\times\cd_y\phi_n\stackrel{L^2}{\wra}0$, whence $\ip{\phi,\phi_x\times\phi_y-\cd_x\phi_n\times\cd_y\phi_n}_{L^2}\ra 0$.
Further, $\phi_n\stackrel{L^2}{\ra}\phi$ and $\|\cd_x\phi_n\times\cd_y\phi_n\|_{L^2}$ is bounded, so 
$\|\phi-\phi_n\|_{L^2}\|\cd_x\phi_n\times\cd_y\phi_n\|_{L^2}\ra 0$ also. Hence $\deg(\phi)=k$.
\end{proof}

The requirement that the potential $V:S^2\ra\R$ have a convex extension to $\R^3$ looks, at first sight, annoyingly restrictive.
In fact, {\em every} function
$S^2\ra\R$ has a convex extension to $\R^3$, provided it is sufficiently smooth, so this is practically no restriction at all:

\begin{prop}\label{convexity} Let $f:S^n\ra\R$ be twice continuously differentiable. Then $f$ has a convex extension $F:\R^{n+1}\ra\R$.
\end{prop}

\begin{proof} Since $f$ is $C^2$ and $S^n$ is compact,
\beq
C_*=\sup\{|(f\circ\gamma)(t)|,|(f\circ\gamma)'(t)|,|(f\circ\gamma)''(t)|\}
\eeq
is finite, where the supremum is over all unit speed geodesics $\gamma:[0,2\pi]\ra S^n$ and all $t\in[0,2\pi]$. 
Choose any $C>2C_*$ and consider the function 
\beq
F:\R^{n+1}\ra\R, \qquad F(x)=|x|^2(f(x/|x|)+C)-C,
\eeq
 which clearly extends $f$. We claim $F$
is convex. To show this it suffices to show that its restriction to any arc-length parametrized straight line in $\R^{n+1}$ is convex. This is clear
for all straight lines through $0$, and for the straight line $\alpha(t)=a+tb$, $a\neq 0$, $|b|=1$, $a\cdot b=0$, one sees, by radially projecting 
$\alpha$ to a geodesic arc on $S^n$,
\beq
\gamma(\theta(t))=\frac{\alpha(t)}{|\alpha(t)|},\qquad \theta(t)=\tan^{-1}\frac{t}{|a|}
\eeq
that
\bea
(F\circ\alpha)''(t)&=&2(f\circ\gamma(\theta(t))+C)+4t(f\circ\gamma)'(\theta(t))\theta'(t)\nonumber\\
&&\quad+
(t^2+|a|^2)(f\circ\gamma)''(\theta(t))\theta'(t)^2
+(t^2+|a|^2)(f\circ\gamma)'(\theta(t))\theta''(t)\nonumber\\
&\geq&\frac{1}{t^2+|a|^2}\left\{2(C-C_*)(t^2+|a|^2)-2t|a|C_*-|a|^2C_*\right\}\nonumber\\
&\geq&\frac{C_*}{t^2+|a|^2}\left\{(t-|a|)^2+t^2\right\}
\geq 0.
\eea
\end{proof}

Combining Theorem \ref{ooo-analysis} and Proposition \ref{convexity}, we see that for any $C^2$ potential $V:S^2\ra\R$ (including $V=0$), any
degree $k$, and any period lattice $\Lambda$, there is a minimizer of $E$ among all $H^1$ maps $\R^2/\Lambda\ra S^2$ of degree $k$. This minimizer is
sufficiently regular for all the integral criteria developed in this section for criticality, and stability, with respect to variations
of $\Lambda$ for this fixed minimizer $\phi$ (that is, the virial constraint, and conformality on average) to be rigorously well-defined. 
There is no reason to expect $H^1$ regularity of the minimizer $\phi$ to be optimal: one would hope, for smooth $V$, that elliptic regularity 
methods would, with some effort, yield smoothness of $\phi$. Since variation of $\phi$ is outside the main focus of the present paper, and enhanced 
regularity is not necessary for our purposes, we do not pursue this further here.

\section{Skyrmion crystals}
\news
\label{sec:skyrme}

In this section we consider the three dimensional Skyrme model. This has a
single scalar field $\phi:M\ra G$, where $(M,g)$ is an oriented riemannian three 
manifold (in our case, $M=\R^3/\Lambda$, a torus), and $G$ is a compact simple 
Lie group (usually taken to be $SU(2)$) whose Lie algebra we denote
$\g$. The energy functional is conventionally\footnote{Actually, Manton and Sutcliffe
take the Skyrme energy to be $E'=\frac{1}{12\pi^2}(2E_2+\frac12 E_4)$ but this can 
be reduced to $E_2+E_4$ by rescaling length and energy units \cite[p.~350]{mansut}.}
the sum of two terms
\beq
E_2=\frac12\int_M|\d\phi|^2\vol_g,\qquad E_4=\frac12\int_M|\phi^*\Omega|^2\vol_g
\eeq
where $\Omega$ is a $\g$-valued two-form on $G$ defined as follows. Let
$\mu\in\Omega^1(G)\otimes\g$ be the left Maurer-Cartan form, that is, the
$\g$-valued one-form on $G$ which associates to any $X\in T_xG$ the value 
at the identity element $e\in G$ of the left invariant vector field on $G$
whose value at $x$ is $X$. Then, for any $X,Y\in T_xG$,
\beq
\Omega(X,Y):=[\mu(X),\mu(Y)].
\eeq
So $\phi^*\Omega\in\Omega^2(M)\otimes\g$, and its norm in the expression for $E_4$
is taken with respect to $g$ and some natural choice of $Ad(G)$ invariant
inner product on $\g$ (for example $\ip{X,Y}_\g=\frac12\tr(X^\dagger Y)$ in the case
of most interest, $G=SU(2)$, giving $G$ the metric of the unit 3-sphere). 
To be explicit, given any local orthonormal frame $e_1,e_2,e_3$ of vector fields
on $M$, then
\beq
|\phi^*\Omega|^2=|\phi^*\Omega(e_1,e_2)|_\g^2+|\phi^*\Omega(e_2,e_3)|_\g^2+
|\phi^*\Omega(e_3,e_1)|_\g^2.
\eeq
One can also allow for the presence of
potential and sextic terms
\beq
E_0=\frac12\int_M U(\phi)^2\vol_g,\qquad
E_6=\frac12\int_M |\phi^*\Xi|^2\vol_g
\eeq
where $U:G\ra\R$ is some potential function and $\Xi\in\Omega^3(G)$ is
some natural three-form on $G$, for example,
\beq
\Xi(X,Y,Z)=\ip{\mu(X),\Omega(Y,Z)}_\g
\eeq
which coincides, in the case $G=SU(2)$, with the volume form on $G$. Such terms
have aroused considerable interest recently because they offer the hope
of constructing so-called ``near-BPS'' Skyrme models with drastically reduced
nuclear binding energies, which addresses a fundamental phenomenological problem
with the usual Skyrme model \cite{adasanwer}. We shall begin our analysis of the model with all
these terms present
\beq
E=E_0+E_2+E_4+E_6
\eeq
before restricting to the usual case by choosing $U=0$, $\Xi=0$. Existence of $H^1$ minimizers in every degree  class on an arbitrary compact domain
for $E=E_2+E_4$, $G=SU(2)$, was established by Kapitanski \cite{kap-compact_skyrme}. A similar result with $E_6$ included (with or without $E_0$)
follows from Proposition \ref{convexity} and the obvious modification of Theorem \ref{ooo-analysis} (the proof of which made no essential use of the
dimension of $M$). Once again, the established regularity is not thought to be optimal, but is sufficient to make the integral criteria below rigorously
well-defined.

For our purposes, the key field theoretic object is the stress tensor of a field
$\phi$. To write this down neatly, we need to generalize slightly the
contraction map $A\cdot B$ for bilinear forms, introduced in section
\ref{sec:var}. So let $A,B$ be $\g$-valued bilinear forms on $M$ (for example,
$\phi^*\Omega$) and $e_1,e_2,e_3$ be a local orthonormal frame of vector
fields on $M$. Then by $A\cdot B$ we will mean the (real valued) bilinear
form
\beq
(A\cdot B)(X,Y)=\sum_i\ip{A(X,e_i),B(e_i,Y)}_\g.
\eeq
With this convention, we have:

\begin{prop} The stress tensor of a Skyrme field $\phi:M\ra G$ with respect to
the energy $E=E_0+E_2+E_4+E_6$ is
$$
S(\phi,g)=\frac14\left(|d\phi|^2+|\phi^*\Omega|^2-|\phi^*\Xi|^2+U(\phi)^2\right)g
-\frac12\left(\phi^*h-\phi^*\Omega\cdot\phi^*\Omega\right).
$$
\end{prop}

\begin{proof}
Let $g_t$ be a smooth variation of $g=g_0$ and $\eps=\cd_t|_{t=0}g_t$. 
The terms coming from $E_0+E_2+E_6$ (that is, all but the second and last terms in the formula
above) were obtained previously \cite{jayspesut}. It remains to show that
\beq
\left.\frac{d\:}{d t}\right|_{t=0}E_4(\phi,g_t)=\ip{\eps,\frac14|\phi^*\Omega|^2g+\frac12\phi^*\Omega\cdot\phi^*\Omega}_{L^2}.
\eeq
Let us employ the abbreviation $\Omega_{ij}=\phi^*\Omega(\cd/\cd x_i,\cd/\cd x_j)$, and the Einstein summation
convention. Then
\beq
|\phi^*\Omega|_g^2=\frac12\ip{\Omega_{ij},\Omega_{kl}}_\g g^{ik}g^{jl}
\eeq
and hence
\beq
\left.\frac{d\: }{dt}\right|_{t=0}|\phi^*\Omega|_{g_t}^2=\eps_{pq}g^{qj}\ip{\Omega_{ji},g^{il}\Omega_{lk}}_\g g^{kp}=\ip{\eps,\phi^*\Omega\cdot\phi^*\Omega}.
\eeq
It follows that
\bea
\left.\frac{d\:}{d t}\right|_{t=0}E_4(\phi,g_t)&=&\left.\frac{d\:}{d t}\right|_{t=0}\frac12\int_M|\phi^*\Omega|_{g_t}^2\vol_{g_t}\nonumber \\
&=&\frac12 \int_M\left(\ip{\eps,\phi^*\Omega\cdot\phi^*\Omega}\vol_{g}+|\phi^*\Omega|^2\left.\frac{d\: }{dt}\right|_{t=0}\vol_{g_t}\right)\nonumber \\
&=&\frac12 \int_M\left(\ip{\eps,\phi^*\Omega\cdot\phi^*\Omega}\vol_{g}+\frac12|\phi^*\Omega|^2\ip{\eps,g}\vol_g\right)
\eea
as required, by (\ref{sopstr}).
\end{proof}

We want to use this formula to extract explicit conditions which $\phi$ must
satisfy if it is to be a skyrmion lattice. Recall that this means precisely that
$S$ is $L^2$ orthogonal to the space $\ee$ of parallel symmetric bilinear forms
on $M$. We expect these conditions to consist of a virial constraint, similar
to (\ref{ijdbs}), and some analogue of the ``conformal on average'' condition
(\ref{injdbs}). To formulate the latter condition in the Skyrme context, we
define $\Delta\in\ee$ as in (\ref{arireb}). That is,
we choose $x\in M$ and define $\Delta:T_xM\times T_xM\ra\R$ by
\beq\label{madkha}
\Delta(X,Y)=\int_M\left(\phi^*h(X,Y)-(\phi^*\Omega\cdot\phi^*\Omega)(X,Y)\right)\vol_g
\eeq
where $X,Y$ on the right are the unique parallel extensions of $X,Y$ over $M$.
We then identify $\Delta$ with an element of $\ee$ using the canonical
isomorphism $\ee\ra T_x^*M\odot T_x^*M$ defined by evaluation. 

\begin{prop}\label{sampen}
 $\phi:M=\R^3/\Lambda\ra G$ is a skyrmion lattice if and only if
$$
(E_2-E_4)+3(E_0-E_6)=0,\qquad\mbox{ and }\qquad\Delta=\frac23(E_2+2E_4)g.
$$
\end{prop}

\begin{proof} To analyze the condition $S\perp_{L^2} g$, we note that, for any symmetric
bilinear form $A$, $\ip{A,g}=\tr A$ pointwise,  and
\beq\label{nausch}
\tr\phi^*h=|\d\phi|^2\qquad\mbox{ and }\qquad \tr\phi^*\Omega\cdot\phi^*\Omega=-2|\phi^*\Omega|^2.
\eeq
Hence
\bea
\ip{S,g}_{L^2}&=&\int_M\left\{
\frac34(|\d\phi|^2+|\phi^*\Omega|^2-|\phi^*\Xi|^2+U(\phi)^2)-
\frac12(|\d\phi|^2+2|\phi^*\Omega|^2)
\right\}\vol_g\nonumber\\
&=&\frac12(E_2-E_4-3E_6+3E_0)
\eea
which establishes the virial constraint. We have already noted that
$S\perp_{L^2}\ee_0$ if and only if $\Delta=\lambda g$ for some constant $\lambda$. 
Taking the trace of both sides and using (\ref{nausch}) again, one finds that
$3\lambda=2E_2+4E_4$.
\end{proof}

Of course, we could have deduced the virial constraint directly from a Derrick
scaling argument, but it is reassuring to see that it follows from our formula
for $S$. 

It would be convenient to have the analogue of Proposition \ref{latimpcry}
for skyrmion lattices, that is, a proof that every skyrmion lattice has
positive hessian. While this is certainly plausible, we have been
unable to prove it because the Skyrme stress tensor lacks a fundamental
simplifying property enjoyed by the baby Skyrme stress tensor. Namely, in the
baby Skyrme case, the derivative of $S$ with respect to $g$ in the
direction of $g$ is itself parallel to $g$, $\dot{S}=\lambda g$, from which it immediately
follows that $\hess$ is block diagonal with respect to the decomposition
$\ee=\ip{g}\oplus\ee_0$. Hence, it suffices to show that $\hess(g,g)>0$ and
$\hess(\eps,\eps)>0$ for all $\eps\in\ee_0\less\{0\}$. 
In the Skyrme case, however, 
$\cd_t|_{t=1} S(\phi,tg)=\lambda g-\frac12\phi^*\Omega\cdot\phi^*\Omega$ (the extra term
coming from the $g$ dependence in the contraction defining 
$\phi^*\Omega\cdot\phi^*\Omega$), and this is not necessarily $L^2$ orthogonal
to $\ee_0$. Hence, the hessian is not block diagonal in general, and no such simplification
occurs.
This difficulty is absent in the special case of Skyrme models for which the
quartic term is absent, that is, with energy $E=E_0+E_2+E_6$, and in this case, 
the analogue of Proposition \ref{latimpcry} does hold (the proof being
essentially unchanged). As far as we are  aware, however, all numerical
studies of Skyrme crystals have addressed the conventional Skyrme model, with
energy $E=E_2+E_4$, so to check whether the solutions found therein are crystals
according to our definition, one must independently check both the lattice 
conditions (given by Proposition \ref{sampen}) and positivity of the hessian.

So, for the rest of this section, we specialize to the usual Skyrme model,
with target space $G=SU(2)$ and energy $E=E_2+E_4$ by setting both
$U$ and $\Xi$ to $0$. Further, we choose $\ip{X,Y}_\g=\frac12\tr(X^\dagger Y)$.
The problem of minimizing $E$ on a cubic torus $T^3_L=\R^3/L\Z^3$ among all Skyrme fields of degree 4 has been studied numerically by
Castillejo {\it et al} 
\cite{casjonjacverjac} and Kugler and Shtrikman \cite{kugsht}. These studies minimized $E$ for fixed side length $L$, and then varied $L$,
independently finding an energy minimum at $L\approx 18.8$ (in our units). The toric 4-skyrmion so obtained is usually called the ``Skyrme crystal''.
Assuming such a minimum exists at this value of $L$, it must satisfy the virial constraint $E_2=E_4$ (since this follows from minimality with respect to 
dilations of $g$). The numerical studies further suggest \cite[p.~383]{mansut}
that the Skyrme crystal is equivariant with respect
to a certain subgroup $K$ of the isometry group of $T^3_L$, which we now 
describe.

\ignore{There is a well-known spatially periodic solution $\phi$ of this model, on a cubic
torus $T^3_L=\R^3/L\Z^3$, found by Castillejo {\it et al} 
\cite{casjonjacverjac} and Kugler and Shtrikman \cite{kugsht}, and usually
called the ``Skyrme crystal.'' Since these numerical studies minimized energy
over the side-length $L$ of the torus (finding a minimum at $L\approx 18.8$ 
in our
units), it is known {\it a priori} that this solution satisfies the virial
constraint $E_2=E_4$. It is also known \cite[p.~383]{mansut}
that $\phi$ is equivariant with respect
to a certain subgroup $K$ of the isometry group of $T^3_L$, which we now 
describe.}

Consider the linear maps $s_i:\R^3\ra\R^3$,
\beq
s_1(x_1,x_2,x_3)=(-x_1,x_2,x_3),\quad
s_2(x_1,x_2,x_3)=(x_2,x_3,x_1),\quad
s_3(x_1,x_2,x_3)=(x_1,x_3,-x_2).
\eeq
Clearly these are isometries of $\R^3$ and preserve any cubic lattice $\Lambda=L\Z^3$, and so generate a subgroup $K$ of $O(3)$ acting
isometrically on $T^3_L$ on the left. This group also acts (isometrically) on $G$ on the left, as follows. We may
identify $G=SU(2)$ with the unit sphere $S^3\subset\R^4$ by means of the correspondence
\beq
(y_0,y_1,y_2,y_3)\leftrightarrow\left(\begin{array}{cc}y_0+iy_2&y_3+iy_1\\
-y_3+iy_1&y_0-iy_2\end{array}\right).
\eeq
This allows us to define an isometric left action of $O(3)$ on $G$ by ${\cal O}:(y_0,\yvec)\mapsto (y_0,{\cal O}\yvec)$, and hence an isometric
left action of $K$ on $G$. The Skyrme crystal is known to be $K$-equivariant with respect to these actions.\footnote{
We have used a slightly non-standard embedding $G\hookrightarrow \R^4$ to define the $K$ action on $G$. This is to ensure that
the Skyrme crystal is $K$-equivariant in the usual sense. Alternatively, we could use the usual embedding, and analyze the Skyrme ``anticrystal'',
$\wt\phi=P\circ\phi$ where $P:\R^4\ra\R^4$ is the map $(y_0,y_1,y_2,y_3)\mapsto (y_0,y_1,y_3,y_2)$.}
We will now show that $K$-equivariance and the virial constraint alone ensure that an $E$ minimizer 
is a soliton crystal according to our definition.
It follows that, assuming the Skyrme crystal exists with the symmetries claimed, it is likewise a soliton crystal
according to our definition.

\begin{prop} Let $\phi:T^3_L\ra SU(2)$ be a $K$-equivariant minimizer of the Skyrme energy $E=E_2+E_4$ which satisfies the virial
constraint $E_2=E_4$. Then $\phi$ is a soliton crystal. 
\end{prop}

\begin{proof} We are given that $\phi$ satisfies the virial constraint, so it suffices to show that $\Delta=\lambda g$, for some
constant $\lambda$ (where $\Delta\in\ee$ is defined in (\ref{madkha})), and that $\hess$ is positive. It is clear that
$E(\phi,g)$ is geometrically natural and $K$-invariant. Hence, by Proposition \ref{sym3}, $\Delta$ is an element of
\beq
\ee^K=\{\eps\in\ee\: :\: \forall k\in K, \: k^*\eps=\eps\},
\eeq
the fixed-point space of $\ee$ under the action of $K$, and by Proposition \ref{sym2}, $\hess$ is
an element of 
\beq
(\ee^*\odot\ee^*)^K=\{H\in\ee^*\odot\ee^*\: :\: \forall k\in K,\: (k^*)^*H=H\}
\eeq 
the fixed-point space of $\ee^*\odot\ee^*$ under the action of $K$.
A representation theoretic argument, presented in the appendix, shows that $\ee^K$ has dimension $1$ and $(\ee^*\odot\ee^*)^K$ has
dimension $3$. Clearly $g\in\ee$ is $K$-invariant, so $\ee^K$ is spanned by $g$, and it follows that $\Delta=\lambda g$
for some $\lambda$. Hence $\phi$ is a soliton lattice. It remains to show that $\hess$ is positive, and for this we introduce a
basis for $(\ee^*\odot\ee^*)^K$ as follows. 

First note that the $K$ action on $\ee$ leaves the three-dimensional subspace of $\ee$
consisting of diagonal symmetric bilinear forms
\beq
\D=\{a_1dx_1^2+a_2dx_2^2+a_3dx_3^2\: :\:  a_1,a_2,a_3\in\R^3\}
\eeq
invariant. It also leaves the line spanned by $g$, a subspace of $\D$, invariant, and preserves the inner product on $\ee$, and hence
leaves the orthogonal complement of $g$ in $\D$  
\beq
\D_0=\{a_1dx_1^2+a_2dx_2^2+a_3dx_3^2\: :\:  a_1,a_2,a_3\in\R^3,\: a_1+a_2+a_3=0\}
\eeq
and the orthogonal complement of $\D$ in $\ee$
\beq
\D^\perp=\{2a_1dx_1dx_2+2a_1dx_1dx_3+2a_3dx_2dx_3\: :\: a_1,a_2,a_3\in\R\}
\eeq
invariant.  Let us denote by $\Pr_g$, $\Pr_{\D_0}$ and $\Pr_{\D^\perp}$ the orthogonal projectors $\ee\ra\ip{g}$, $\ee\ra\D_0$ and
$\ee\ra\D^\perp$, and by $H_1,H_2,H_3:\ee\times\ee\ra\R$ the symmetric bilinear forms
\bea
H_1(\wh\eps,\eps)&=&\ip{\Pr_g(\wh\eps),\Pr_g(\eps)}_\ee,\nonumber \\
H_2(\wh\eps,\eps)&=&\ip{\Pr_{\D_0}(\wh\eps),\Pr_{\D_0}(\eps)}_\ee,\nonumber\\
\label{flablaski}
H_3(\wh\eps,\eps)&=&\ip{\Pr_{\D^\perp}(\wh\eps),\Pr_{\D^\perp}(\eps)}_\ee.
\eea
These are linearly independent and, by construction, $K$-invariant, and hence form a basis for $(\ee^*\odot\ee^*)^K$.
Hence
\beq
\hess=c_1H_1+c_2H_2+c_3H_3
\eeq
for some constants $c_1,c_2,c_3$, and $\hess$ is positive if and only if these constants are positive. 

Now 
\bea
3c_1&=&\hess(g,g)=\left.\frac{d^2\: }{dt^2}\right|_{t=1}E(\phi,tg)
=\left.\frac{d^2\: }{dt^2}\right|_{t=1}(t^{\frac12}E_2(\phi,g)+t^{-\frac12}E_4(\phi,g))\nonumber \\
&=&\frac14E_2+\frac34E_4=E_2>0.
\eea
To compute $c_2$ and $c_3$, it is convenient to use again the
abbreviation
$\Omega_{ij}:=\phi^*\Omega(\cd/\cd x_i,\cd/\cd x_j)$. Let $g_t$ be a generating curve in $\ee$ for $\eps$.
Then
\bea
\dot S&=&\left.\frac{d\: }{dt}\right|_{t=0}S(\phi,g_t)=\lambda g+\frac14(|\d\phi|^2+|\phi^*\Omega|^2)\eps+\frac12\left.\frac{d\: }{dt}\right|_{t=0}
(\phi^*\Omega)\cdot_{g_t}(\phi^*\Omega)\nonumber \\
&=&\lambda g+\frac14(|\d\phi|^2+|\phi^*\Omega|^2)\eps
-\frac12\sum_{i,j,k,l}\ip{\Omega_{ik},\eps_{kl}\Omega_{lj}}_\g\, dx_idx_j.
\eea
In the case $\eps=dx_1^2-dx_2^2\in\D_0$, one sees that
\beq
\dot S=\frac14(|\d\phi|^2+|\phi^*\Omega|^2)\eps-\frac12|\Omega_{12}|^2\eps+A
\eeq
where $\ip{A,\eps}_\ee=0$. Hence, by Proposition \ref{hessprop},
\bea
2c_2&=&\hess(dx_1^2-dx_2^2,dx_1^2-dx_2^2)
=\ip{\dot S,\eps}_{L^2}
=E_2+E_4-\int_M|\Omega_{12}|^2\vol_g\nonumber \\
&=&2E_4-\int_M|\Omega_{12}|^2\vol_g
=\|\Omega_{23}\|_{L^2}^2+\|\Omega_{31}\|_{L^2}^2>0.
\eea
Similarly, in the case $\eps=2dx_1dx_2\in\D^\perp$,
\beq
\dot S=\frac14(|\d\phi|^2+|\phi^*\Omega|^2)\eps-\frac12|\Omega_{12}|^2\eps+A'
\eeq
where $\ip{A',\eps}_\ee=0$ and so, by identical reasoning, 
\beq
2c_3=\hess(2dx_1dx_2,2dx_1dx_2)=\|\Omega_{23}\|_{L^2}^2+\|\Omega_{31}\|_{L^2}^2>0.
\eeq
\end{proof}

Other periodic Skyrme solutions on cubic tori have been found numerically, and can be analyzed in similar fashion.
For example, Klebanov \cite{kle} found a solution which is equivariant with respect to the
subgroup $K'<K$ generated by $s_1,s_2$ only. Again, he minimized over period length, so that the virial constraint 
must hold.
It turns out that $\ee^{K'}=\ee^K$ and $(\ee^*\odot\ee^*)^{K'}=(\ee^*\odot\ee^*)^{K}$,
so exactly the same argument given above shows that this toric soliton is also a soliton crystal.

It is slightly surprising that, in the course of the proof above, we showed that $c_2=c_3$, so that the hessian
actually has the simple form
\beq
\hess(\wh\eps,\eps)=c_1\ip{\Pr_g(\wh\eps),\Pr_g(\eps)}_\ee+c_2\ip{\Pr_{\ee_0}(\wh\eps),\Pr_{\ee_0}(\eps)}_\ee.
\eeq
This does not follow immediately from $K$-equivariance alone, but seems to rely on the detailed structure of the Skyrme
energy, so should not be expected as a generic property of soliton lattices on cubic tori. Note that the symmetry $\phi\circ s_2=s_2\circ\phi$
implies 
\beq
\|\Omega_{12}\|_{L^2}^2=\|\Omega_{23}\|_{L^2}^2=\|\Omega_{31}\|_{L^2}^2=\frac23 E_4=\frac23 E_2
=\frac13E,
\eeq
 so that $K$-equivariance
(or $K'$-equivariance) actually implies
\beq
\hess(\wh\eps,\eps)=\ignore{\frac13E_2(\ip{\Pr_g(\wh\eps),\Pr_g(\eps)}_\ee+2\ip{\Pr_{\ee_0}(\wh\eps),\Pr_{\ee_0}(\eps)}_\ee)
=}
\frac13E\left(\ip{\wh\eps,\eps}_\ee-\frac16\ip{g,\wh\eps}_\ee\ip{g,\eps}_\ee\right)
\eeq
for the standard Skyrme model. It would be interesting to see whether useful information about the vibrational modes of
Skyrme crystals can be extracted from this formula.

\section{Concluding remarks}
\news
\label{sec:conc}
We have derived necessary conditions for a soliton on a torus $M=\R^m/\Lambda$ to be a soliton crystal. The
stress tensor $S$ of the soliton must be $L^2$ orthogonal to $\ee$, the space of parallel symmetric bilinear
forms on $TM$ and, further, a certain symmetric bilinear form on $\ee$, called the hessian, must be positive.
We have shown that, for baby Skyrme models, the first condition actually implies the second. We have also shown that,
for any choice of lattice $\Lambda$, there is a baby Skyrme model which supports a soliton crystal of periodicity $\Lambda$.
For the three-dimensional Skyrme model, we showed that a soliton solution on a cubic lattice which
satisfies the virial constraint $E_2=E_4$ and is equivariant with respect to (a subgroup of) the lattice
symmetries automatically satisfies both tests. This verifies in particular that the ``Skyrme
crystal'' of Castillejo {\it et al.}, and Kugler and Shtrikman, passes both tests. 
Note that, although we have applied the criteria only to local minimizers of $E(\phi)$, they could equally well be applied to saddle points (unstable
static solutions), and there is no obvious reason why a saddle point for variations of $\phi$ should not be a local minimum for variations of $\Lambda$.
Indeed, {\em if} saddle points of the baby Skyrme energy exist which are critical for variations of $\Lambda$, the proof of Proposition \ref{latimpcry} 
implies that they can {\em only} be local minima (with respect to variations of $\Lambda$). The physical significance of such saddle points, if, indeed, they
exist at all, is not clear.

It would be straightforward to extend the analysis to deal with gauge theories on tori. In this context, $\phi$ would be a section of
some vector bundle $V$ over $M$ with connexion $\nabla$. The linear diffeomorphisms used to identify
all tori with $M$ can be used to identify the bundles $V$, the section $\phi$, and the connexion $\nabla$, by pullback, so that, once again, the variation
over period lattices is reformulated as a variation over metrics. One expects the first and second variation formulae to be structurally identical
to those presented in section \ref{sec:var}, therefore.

A more interesting extension would be to apply the idea to soliton sheets and chains, that is, solitons on $M=\R^m\times T^{m'}$. Presumably
the first variation formula will, once again, amount to the condition that the soliton satisfies all generalized Derrick constraints on $M$. Now, however,
$\frac12m(m+1)$ of these conditions will be known {\it a priori} for free (or $\frac12m(m+1)+1$ if the energy has been minimized over torus volume, as is usual
in numerical studies) because they involve deformations only of the $\R^m$ factor, which are already accounted for in the variation of $\phi$. This leaves
$\frac12m'(m'+1)+mm'$ nontrivial constraints. Equivariance with respect to lattice symmetries is likely to be considerably less constraining than it is for
tori so that equivariance alone is unlikely to guarantee criticality. It would be interesting to see whether the skyrmion sheets found numerically in
\cite{batsut-hex} and \cite{silwar} survive the analysis.

\appendix
\section*{Appendix: the symmetry group $K$ and its action on $\ee$}

\news
\renewcommand{\theequation}{A.\arabic{equation}}

Recall that the Skyrme crystal is equivariant with respect to a discrete subgroup $K<O(3)$ generated by the matrices
\beq
s_1=\left(\begin{array}{ccc}-1&0&0\\0&1&0\\0&0&1\end{array}\right),\quad 
s_2=\left(\begin{array}{ccc}0&1&0\\0&0&1\\1&0&0\end{array}\right),\quad 
s_3=\left(\begin{array}{ccc}1&0&0\\0&0&1\\0&-1&0\end{array}\right).
\eeq
These generate a nonabelian group of order
$48$,  consisting of all permutation matrices where each nonzero entry can be either $1$ or $-1$.
 There is an induced isometric action of $K$ on
$(\ee,\ip{\cdot,\cdot}_\ee)$ by pullback, 
\beq
k^*\eps(X,Y)=\eps(k(X),k(Y))
\eeq
 and, further, an induced action of $K$ on $\ee^*\odot\ee^*$ by pullback of the pullback,
\beq
(k^*)^*H(\wh\eps,\eps)=H(k^*\wh\eps,k^*\eps).
\eeq
We wish to compute the dimensions of their fixed point spaces $\ee^K$, $(\ee^*\odot\ee^*)^K$, or equivalently, to count
the number of copies of the trivial representation in the decomposition of the $K$ representations on $\ee$ and $\ee^*\odot\ee^*$
into irreducible orthogonal representations. This
we can do by using character orthogonality. Our task, therefore, boils down to the construction of character tables for these two representations.

Recall that characters are constant on conjugacy classes. There is an obvious eight-to-one
``forgetful'' homomorphism $\mu:K\ra S_3$, which sends each signed permutation matrix to the permutation
matrix obtained by changing all $-1$ entries to $+1$. Clearly, if $k,k'\in K$ are conjugate in $K$, so are
$\mu(k),\mu(k')$ is $S_3$, so each conjugacy class $[k]$ in $K$ carries a label $[\mu(k)]$, a conjugacy class in $S_3$.
There are three such classes, consisting of permutations which fix 3, 1 or 0 elements, the classes of
$e$, $(23)$ and $(132)$ respectively. Conjugate elements in $K$ also have equal trace and determinant. 
So if $k$ is conjugate to $k'$, $([\mu(k)],\det k, \tr k)=([\mu(k')],\det k', \tr k')$ and straightforward calculation
shows that the converse also holds: if $([\mu(k)],\det k,\tr k)=([\mu(k')],\det k', \tr k')$ then $k$ is conjugate to $k'$.
Hence, each conjugacy class is uniquely labelled by the triple $([\mu(k)],\det k, \tr k)$. From this we deduce that $K$ splits
into 10 conjugacy classes, as specified in table \ref{table1}. The final label, $\tr k$, is also the character $\chi^{\R^3}(k)$ of the
fundamental representation of $K$. From this, we can deduce the character of the induced representation on $\ee$, the space of
symmetric bilinear forms on $\R^3$, using the standard formula \cite{jon}
\beq\label{heauga}
\chi^\ee(k)=\frac12\left[\chi^{\R^3}(k)^2+\chi^{\R^3}(k^2)\right].
\eeq
For this purpose we need to know $[k^2]$, the conjugacy class of $k^2$, for a representative $k$ of each class. This information is recorded
in column 6 of table \ref{table1}, and suffices to compute $\chi^\ee(k)$ for each class, yielding column 7.\, To compute the number of copies of the trivial 
representation of $K$ in $\chi^\ee$, we compute the character inner product between  $\chi^\ee$ and $\chi^{triv}$ (where $\chi^{triv}(k)=1$ for all $k$):
\beq
\ip{\chi^\ee,\chi^{triv}}=\frac{1}{|K|}\sum_{k\in K}\chi^\ee(k)\chi^{triv}(k)=
\frac{1}{|K|}\sum_{[k]\in K/\sim}|[k]|\chi^\ee([k])\times 1=1
\eeq
where the second sum is over conjugacy classes. Hence, $\ee^K$ is one-dimensional. 

\begin{table}
$$
\begin{array}{c|ccccccc}
\begin{array}{c}\mbox{Representative}\\k\end{array}& |[k]| & [\mu(k)] & \det k & \tr k=\chi^{\R^3}(k) & [k^2] & \chi^\ee(k) & \chi^{\ee^*\odot\ee^*}(k) \\ \hline
\I_3 & 1 & [e] & 1 & 3 & [\I_3] & 6 & 21 \\
-\I_3 & 1 & [e] & -1 & -3 & [\I_3] & 6 & 21 \\
\downsize{\left(\begin{array}{ccc}-1&\o&\o\\\o&1&\o\\\o&\o&1\end{array}\right)} & 3 & [e] & -1 & 1 & [\I_3]& 2 & 5 \\
\downsize{\left(\begin{array}{ccc}-1&\o&\o\\\o&-1&\o\\\o&\o&1\end{array}\right)} & 3 & [e] & 1 & -1 & [\I_3]& 2 & 5 \\
\downsize{\left(\begin{array}{ccc}1&\o&\o\\\o&\o&1\\\o&1&\o\end{array}\right)} & 6 & [(2,3)] & -1 & 1 & [\I_3]& 2 & 5 \\
\downsize{\left(\begin{array}{ccc}-1&\o&\o\\\o&\o&-1\\\o&-1&\o\end{array}\right)} & 6 & [(2,3)] & 1 & -1 & [\I_3]& 2 & 5 \\
\downsize{\left(\begin{array}{ccc}1&\o&\o\\\o&\o&-1\\\o&1&\o\end{array}\right)} & 6 & [(2,3)] & 1 & 1 & \downsize{\left[\left(\begin{array}{ccc}-1&\o&\o\\\o&-1&\o\\\o&\o&1\end{array}\right)\right]}& 0 & 1 \\
\downsize{\left(\begin{array}{ccc}-1&\o&\o\\\o&\o&1\\\o&-1&\o\end{array}\right)} & 6 & [(2,3)] & -1 & -1 & \downsize{\left[\left(\begin{array}{ccc}-1&\o&\o\\\o&-1&\o\\\o&\o&1\end{array}\right)\right]}& 0 & 1 \\
\downsize{\left(\begin{array}{ccc}\o&1&\o\\\o&\o&1\\1&\o&\o\end{array}\right)} & 8 & [(1,3,2)] & 1 & 0 & \downsize{\left[\left(\begin{array}{ccc}\o&1&\o\\\o&\o&1\\1&\o&\o\end{array}\right)\right]}& 0 & 0 \\
\downsize{\left(\begin{array}{ccc}\o&-1&\o\\\o&\o&-1\\-1&\o&\o\end{array}\right)} & 8 & [(1,3,2)] & -1 & 0 & \downsize{\left[\left(\begin{array}{ccc}\o&1&\o\\\o&\o&1\\1&\o&\o\end{array}\right)\right]}& 0 & 0
\end{array}
$$
\caption{\sf Character table for the representations of the signed permutation group $K$ on $\R^3$, $\ee$ and $\ee^*\odot\ee^*$. The columns, from left to right,
give a representative $k$ for each conjugacy class, the size of the class, then the three labels $([\mu(k)],\det k, \tr k)$ which uniquely label the
class. The last of these coincides with the character $\chi^{\R^3}$. The next column gives the class of $k^2$, which suffices to compute, inductively,
$\chi^\ee$ and $\chi^{\ee^*\odot\ee^*}$ using the formulae (\ref{heauga}) and (\ref{blaski}).}
\label{table1}
\end{table}

Consider now the induced representation of $K$ on $\ee^*\odot\ee^*$, the 21-dimensional space of symmetric bilinear forms on $\ee$. This is just the adjoint
representation associated with the orthogonal representation of $K$ on $\ee$ just constructed. Hence, its character is related to $\chi^\ee$
just as in (\ref{heauga}), namely
\beq\label{blaski}
\chi^{\ee^*\odot\ee^*}(k)=\frac12\left[\chi^{\ee}(k)^2+\chi^{\ee}(k^2)\right],
\eeq
which yields column 8 of table \ref{table1}. Now
\beq
\ip{\chi^{\ee^*\odot\ee^*},\chi^{triv}}=
\frac{1}{|K|}\sum_{[k]\in K/\sim}|[k]|\chi^{\ee^*\odot\ee^*}([k])=3
\eeq
so we deduce that $(\ee^*\odot\ee^*)^K$ has dimension $3$.

A similar analysis can be performed for $K'$, the order $24$ group generated by $s_1,s_2$ alone. One finds that the spaces $\ee^{K'}$ and $(\ee^*\odot\ee^*)^{K'}$ again
have dimension $1$ and $3$ respectively. Since $K'<K$, it follows that $\ee^K<\ee^{K'}$ and $(\ee^*\odot\ee^*)^K<(\ee^*\odot\ee^*)^{K'}$ and hence, by equality of dimensions,
 $\ee^K=\ee^{K'}$, $(\ee^*\odot\ee^*)^K=(\ee^*\odot\ee^*)^{K'}$.

\section*{Acknowledgements}
This work was partially funded by the UK Engineering and Physical Sciences
Research Council.


\begin{thebibliography}{10}
\newcommand{\enquote}[1]{``#1''}

\bibitem{adasanwer}
C.~{Adam}, J.~{S{\'a}nchez-Guill{\'e}n} and A.~{Wereszczy{\'n}ski}, \enquote{{A
  Skyrme-type proposal for baryonic matter}}, {\em Phys.\ Lett.\/} {\bf B691}
  (2010), 105--110.

\bibitem{ada}
R.~A. Adams, {\em Sobolev Spaces\/} (Academic Press, London U.K., 1975).

\bibitem{baieel}
P.~Baird and J.~Eells, \enquote{A conservation law for harmonic maps}, in {\em
  Geometry {S}ymposium, {U}trecht 1980 ({U}trecht, 1980)\/}, vol. 894 of {\em
  Lecture Notes in Math.\/} (Springer, Berlin, 1981), pp. 1--25.

\bibitem{baiwoo}
P.~Baird and J.~C. Wood, {\em Harmonic morphisms between {R}iemannian
  manifolds\/}, vol.~29 of {\em London Mathematical Society Monographs. New
  Series\/} (Oxford University Press, Oxford, U.K., 2003).

\bibitem{batsut-hex}
R.~A. {Battye} and P.~M. {Sutcliffe}, \enquote{{A Skyrme lattice with hexagonal
  symmetry}}, {\em Phys.\ Lett.\/} {\bf B416} (1998), 385--391.

\bibitem{casjonjacverjac}
L.~{Castillejo}, P.~S.~J. {Jones}, A.~D. {Jackson}, J.~J.~M. {Verbaarschot} and
  A.~{Jackson}, \enquote{{Dense skyrmion systems}}, {\em Nucl.\ Phys.\/} {\bf
  A501} (1989), 801--812.

\bibitem{der}
G.~H. Derrick, \enquote{{Comments on nonlinear wave equations as models for
  elementary particles}}, {\em J. Math.\ Phys.\/} {\bf 5} (1964), 1252--1254.

\bibitem{domhoyson}
S.~K. Domokos, C.~Hoyos and J.~Sonnenschein, \enquote{{Deformation Constraints
  on Solitons and D-branes}}, {\em arXiv:1306.0789\/}  (2013).

\bibitem{eva}
L.~C. Evans, {\em Weak convergence methods for nonlinear partial differential
  equations\/}, vol.~74 of {\em CBMS Regional Conference Series in
  Mathematics\/} (Published for the Conference Board of the Mathematical
  Sciences, Washington, DC, 1990).

\bibitem{jayspesut}
J.~{Jaykka}, M.~{Speight} and P.~{Sutcliffe}, \enquote{{Broken baby
  Skyrmions}}, {\em Proc.\ R. Soc. Lond.\/} {\bf A468} (2012), 1085--1104.

\bibitem{jon}
H.~F. Jones, {\em Groups, Representations and Physics\/} (Adam Hilger, Bristol
  U.K., 1990).

\bibitem{kap-compact_skyrme}
L.~Kapitanski, \enquote{On {S}kyrme's model}, in {\em Nonlinear problems in
  mathematical physics and related topics, {II}\/}, vol.~2 of {\em Int. Math.
  Ser. (N. Y.)\/} (Kluwer/Plenum, New York, 2002), pp. 229--241.

\bibitem{henkar}
M.~Karliner and I.~Hen, \enquote{Rotational symmetry breaking in baby {S}kyrme
  models}, in G.~E. Brown and M.~Rho (editors), {\em The Multifaceted
  Skyrmion\/} (World Scientific, Singapore, 2010), pp. 179--213.

\bibitem{kle}
I.~{Klebanov}, \enquote{{Nuclear matter in the skyrme model}}, {\em Nucl.\
  Phys.\/} {\bf B262} (1985), 133--143.

\bibitem{kugsht}
M.~{Kugler} and S.~{Shtrikman}, \enquote{{A new skyrmion crystal}}, {\em Phys.\
  Lett.\/} {\bf B208} (1988), 491--494.

\bibitem{law}
D.~F. Lawden, {\em Elliptic Functions and Applications\/} (Springer-Verlag,
  London U.K., 1989).

\bibitem{lic}
A.~Lichnerowicz, \enquote{Applications harmoniques et vari\'et\'es
  k\"ahleriennes}, {\em Symp.\ Math.\ Bologna\/} {\bf 3} (1970), 341--402.

\bibitem{linyan-bs}
F.~Lin and Y.~Yang, \enquote{Existence of two-dimensional skyrmions via the
  concentration-compactness method}, {\em Comm. Pure Appl. Math.\/} {\bf 57}
  (2004), 1332--1351.

\bibitem{man-der}
N.~S. {Manton}, \enquote{{Scaling identities for solitons beyond Derrick's
  theorem}}, {\em J. Math.\ Phys.\/} {\bf 50} (2009), 032901.

\bibitem{mansut}
N.~S. Manton and P.~M. Sutcliffe, {\em Topological Solitons\/} (Cambridge
  University Press, Cambridge U.K., 2004).

\bibitem{mor}
T.~J. Morrison, {\em Functional Analysis. An Introduction to Banach Space
  Theory\/} (John Wiley and Sons, New York N.Y., U.S.A., 2001).

\bibitem{silwar}
J.~Silva~Lobo and R.~S. Ward, \enquote{Skyrmion multi-walls}, {\em J. Phys.\/}
  {\bf A42} (2009), 482001.

\bibitem{spe-semicompactons}
J.~M. {Speight}, \enquote{{Compactons and semi-compactons in the extreme baby
  Skyrme model}}, {\em J. Phys.\/} {\bf A43} (2010), 405201.

\bibitem{war}
R.~S. Ward, \enquote{Planar {S}kyrmions at high and low density}, {\em
  Nonlinearity\/} {\bf 17} (2004), 1033--1040.

\end{thebibliography}
\end{document}